\documentclass{IEEEtran}

\usepackage[english]{babel}
\usepackage[utf8]{inputenc}
\usepackage{amssymb,stmaryrd, amsmath, amsthm}
\usepackage{Prooftree}
\usepackage[all]{xy}
\usepackage{graphicx}

\newcommand{\restrict}{\upharpoonright}

\newcommand{\tto}{\Rightarrow}

\newcommand{\intr}[1]{\llbracket #1 \rrbracket}
\newcommand{\lpair}[1]{\langle #1 \rangle}

\newcommand{\leg}[1]{\mathcal{L}_{#1}}

\newcommand{\one}{\mathbf{1}}

\newcommand{\unfold}[1]{\widetilde{A}}

\newcommand{\Fam}{\mathrm{Fam_f}}
\newcommand{\BFam}{\mathrm{Fam}}
\newcommand{\enb}[1]{\mathrel{\vdash_{#1}}}
\newcommand{\var}[1]{\mathtt{var}[#1]}
\newcommand{\sleq}{\preccurlyeq}
\newcommand{\comp}{\mathrm{comp}}
\newcommand{\obleq}{\leq}
\newcommand{\obeq}{\cong}
\newcommand{\path}[1]{\mathcal{P}_{#1}}
\newcommand{\threads}[1]{\mathcal{T}_{#1}}
\newcommand{\prethreads}[1]{\threads{#1}'}
\newcommand{\ip}{\mathrm{ip}}
\newcommand{\jp}{\mathrm{jp}}
\newcommand{\iso}{\simeq}
\newcommand{\eval}{\Downarrow}
\newcommand{\obs}{\cong}
\newcommand{\thread}[1]{\lceil #1 \rceil} 
\newcommand{\preleg}[1]{\leg{A}'}

\newcommand{\C}{\mathcal{C}}
\newcommand{\D}{\mathcal{D}}

\newcommand{\Gam}{\mathbf{Gam}}
\newcommand{\biggam}{\mathbf{Gam}_\infty}
\newcommand{\Vis}{\mathbf{Vis}}
\newcommand{\Inn}{\mathbf{Inn}}
\newcommand{\Path}{\mathbf{Path}}
\newcommand{\Seq}{\mathbf{Seq}}
\newcommand{\Jus}{\mathbf{Jus}}
\newcommand{\Set}{\mathbf{Set}}

\newcommand{\Lang}{\mathcal{L}}

\newtheorem{lemma}{\textsc{Lemma}}
\newtheorem{theorem}{\textsc{Theorem}}
\newtheorem{definition}{\textsc{Definition}}
\newtheorem{proposition}{\textsc{Proposition}}

\newlength{\viewht}
\newlength{\viewlift}
\newlength{\viewdp}
\newlength{\viewdrop}

\newcommand{\pview}[1]{
\settoheight{\viewht}{\makebox{$#1$}}
\setlength{\viewlift}{\viewht}%
\addtolength{\viewlift}{-1ex}%
\raisebox{0.3\viewlift}{
  \makebox{$\ulcorner$}}
  \!#1\!
\settoheight{\viewht}{\makebox{$#1$}}
\setlength{\viewlift}{\viewht}%
\addtolength{\viewlift}{-1ex}%
\raisebox{0.3\viewlift}{
  \makebox{$\urcorner$}}
}

\begin{document}
\title{Isomorphisms of types\\ in the presence of higher-order references}

\author{\IEEEauthorblockN{Pierre Clairambault}\\
University of Bath\\
\texttt{p.clairambault@bath.ac.uk}}
\maketitle

\begin{abstract}
We investigate the problem of type isomorphisms in a programming language with higher-order references. We first recall
the game-theoretic model of higher-order references by Abramsky, Honda and McCusker. Solving an open problem by Laurent,
we show that two finitely branching arenas are isomorphic if and only if they are geometrically the same, up to renaming of moves
(Laurent's forest isomorphism). We deduce from this an equational theory characterizing isomorphisms of types in a finitary language
$\Lang_2$ with higher order references. We show however that Laurent's conjecture does not hold on infinitely branching arenas,
yielding a non-trivial type isomorphism in the extension of $\Lang_2$ with natural numbers.
\end{abstract}

\section{Introduction}

During the development of denotational semantics of programming languages, there was a crucial interest in defining models of computation
satisfying particular type equations. For instance, a model of the untyped $\lambda$-calculus can be obtained by isolating a \emph{reflexive} object
(that is, an object $D$ such that $D\iso D^D$) in a cartesian closed category. In the 80s, some people started to consider the dual problem
of finding these equations that must hold in \emph{every} model of a given language: they were coined \emph{type isomorphisms} by Bruce and Longo. In \cite{DBLP:conf/stoc/BruceL85}, 
they exploited a theorem by Dezani \cite{dezani1976characterization} giving a syntactic characterization of invertible terms in the untyped $\lambda$-calculus
to prove that that the only isomorphisms of types present in simply typed $\lambda$-calculus with respect to $\beta\eta$ equality are those induced by 
the equation $A\to (B\to C) \iso B\to (A\to C)$. Later this was extended to handle such things as products \cite{DBLP:journals/mscs/BruceCL92}, higher order \cite{DBLP:conf/stoc/BruceL85},
possibly with unit types \cite{diinvertibility}, or sums \cite{DBLP:journals/apal/FioreCB06}.

The interest in type isomorphisms grew significantly when their practical impact was realized. In \cite{DBLP:journals/jfp/Rittri91}, Rittri proposed to search functions
in software libraries using their type modulo isomorphism as a key. He also considered the possibilities offered by matching and unification of types
modulo isomorphisms \cite{DBLP:journals/ita/Rittri93}. A whole line of research has also been dedicated to the study of type isomorphisms and their use for search tools in richer type systems
(such as dependent types \cite{DBLP:conf/fossacs/BartheP01}), along with studies about the automatic generation of the corresponding coercions \cite{DBLP:conf/mpc/AtanassowJ04}.
Such tools were implemented for several programming languages, let us mention the command line tool
\texttt{camlsearch} written by Vouillon for CamlLight.

It is worth noting that even though these tools are written for powerful programming languages featuring
complex computational effects such as higher-order references or exceptions, they rely on the theory of isomorphisms in weaker (purely functional) languages, such as the second-order $\lambda$-calculus
with pairs and unit types for \texttt{camlsearch}. Clearly, all type isomorphisms in $\lambda$-calculus are still valid in the presence of computational effects (indeed, the operational semantics are compatible with $\beta\eta$). What is
less clear is whether those effects allow the definition of new isomorphisms. However, it seems that syntactic methods deriving from Dezani's theorem on invertible terms
in $\lambda$-calculus cannot be extended to complex computational effects. The base setting itself is completely different: the dynamics of terms are no longer defined by reduction
rules but by operational semantics, the natural equality between terms is no longer convertibility but observational equivalence, so new methods are required.

In \cite{DBLP:journals/mscs/Laurent05}, Laurent introduced the idea of applying game semantics to the study of type isomorphisms (although one should
mention the precursor characterization of isomorphisms by Berry and Curien \cite{berry-curien} in the category of concrete data structures and sequential algorithms).
Exploiting his earlier work on game semantics for polarized linear logic \cite{DBLP:journals/apal/Laurent04}, he found the theory of isomorphisms for LLP from
which he deduced (by translations) the isomorphisms for the call-by-name and call-by-value $\lambda\mu$-calculus. The core of his analysis is the
observation that isomorphisms between arenas $A$ and $B$ in the category $\Inn$ \cite{hyland-ong} of arenas and innocent strategies are in one-to-one correspondence with
\emph{forest isomorphisms} between $A$ and $B$, so in particular two arenas are isomorphic if and only if their representations as forests are identical up to the renaming of vertices.

From the point of view of computational effects this looks promising, since game semantics are known to accommodate several computational effects such as control operators \cite{laird97},
ground type \cite{abramsky-mccusker:active-algol} or higher-order references \cite{ahm} or even concurrency \cite{DBLP:conf/concur/Laird05} in one single framework. Moreover,
Laurent pointed out in \cite{DBLP:journals/mscs/Laurent05} that the main part of his result, namely the fact that each $\Inn$-isomorphism induces a forest isomorphism, does not really depend on the 
innocence hypothesis but only on the weaker \emph{visibility} condition. As a consequence, his method for characterizing isomorphisms still applies to programming languages such as Idealized Algol whose
terms can be interpreted as visible strategies \cite{abramsky-mccusker:active-algol}. Laurent raised the question whether his result could be proved without the visibility condition,
therefore yielding a characterization of isomorphisms in a programming language whose terms have access to higher-order references and hence get interpreted as non-visible strategies \cite{ahm}.

The contribution of this paper is threefold: \emph{(1)} We give a new and synthetic reformulation of Laurent's tools to approach game-theoretically the problem of type isomorphisms,
\emph{(2)} We prove Laurent's conjecture in the case of finitely branching arenas, allowing us to characterize all type isomorphisms in a finitary (integers-free)
programming language $\Lang_2$ with higher-order references by the theory $\mathcal{E}$ presented\footnote{The absence of the equation $A\to (B\to C) \iso B \to (A\to C)$ mentioned in the introduction may seem strange, but is standard in call-by-value \cite{DBLP:journals/mscs/Laurent05} due to the restriction of the
$\eta$-rule on values.} in Figure \ref{equational_theory}, \emph{(3)} We show however a counter-example to
the conjecture when dealing with infinitely branching arenas, and the counter-example yields a non-trivial type isomorphism
between the types $(\mathtt{nat} \to \mathtt{unit}) \to (\mathtt{nat} \to \mathtt{unit}) \to \mathtt{unit}$ and
$(\mathtt{nat} \to \mathtt{unit}) \to (\mathtt{unit} \to \mathtt{unit}) \to \mathtt{unit}$
in the extension of $\Lang_2$ with natural numbers. So Laurent's conjecture, in the general case, is false.

In Section \ref{section_lang} we introduce the finitary language $\Lang_2$ strongly inspired by Abramsky, Honda and McCusker's language $\Lang$ \cite{ahm}, along with
its standard game semantics. Then we turn to the problem of isomorphisms of types. In Section \ref{section_isomorphisms} we first give an analysis of isomorphisms
in several subcategories of the games model, reproving and extending Laurent's theorem, then we use it to characterize isomorphisms in $\Lang_2$. We show how this
characterization fails in the presence of natural numbers, and we give a non-trivial type isomorphism in $\Lang$.

\begin{figure}
\begin{eqnarray*}
A\times B &\iso_{\mathcal{E}} & B \times A\\
A\times (B\times C) &\iso_{\mathcal{E}}& (A\times B)\times C\\
A\times \mathtt{unit} &\iso_{\mathcal{E}}& A\\
\mathtt{bool}\times A \to B &\iso_{\mathcal{E}}& (A\to B) \times (A\to B)\\
\mathtt{var}[A] &\iso_\mathcal{E}& (A\to \mathtt{unit})\times (\mathtt{unit} \to A)
\end{eqnarray*}
\caption{Isomorphisms in $\Lang_2$}
\label{equational_theory}
\end{figure}

\section{The language $\Lang_2$ and its game semantics}
\label{section_lang}
\subsection{Definition of $\Lang_2$}
\paragraph{Basic definitions}
We introduce here a finitary variant $\Lang_2$ of the programming language $\Lang$ with higher-order references modeled by Abramsky, Honda and McCusker in \cite{ahm}:
it only differs from $\Lang$ in the fact that the type of natural numbers has been replaced with a type for booleans, along with all the associated combinators.
The terms and types of $\Lang_2$ are given by the following grammars.
\begin{eqnarray*}
A&::=& \mathtt{unit}~|~\mathtt{bool}~|~A\times A~|~A\to A~|~\mathtt{var}[A]\\\\
M &::=& x~|~\lambda x.M~|~M~M~|~\lpair{M, M}~|~\mathtt{fst}~M~|~\mathtt{snd}~M\\
&&|~\mathtt{skip}~|~\mathtt{true}~|~\mathtt{false}~|~\mathtt{if}~M~M~M\\
&&|~\mathtt{new}_A~|~M:=M~|~!M~|~\mathtt{mkvar}~M~M
\end{eqnarray*}
The typing rules for $\lambda$-calculus, pairs and booleans are standard. The rules for references follow.
\[
\prooftree
        \justifies
        \Gamma \vdash \mathtt{new}_A : \mathtt{var}[A]
\endprooftree
~~~~~~~~~~~~
\prooftree
        \Gamma \vdash M:\mathtt{var}[A]
        \justifies
        \Gamma\vdash !M : A
\endprooftree
\]
\[
\prooftree
        \Gamma \vdash M:\mathtt{var}[A]~~~~\Gamma\vdash N:A
        \justifies
        \Gamma\vdash M:=N : \mathtt{unit}
\endprooftree
\]
\[
\prooftree
        \Gamma\vdash M:A \to \mathtt{unit}~~~~\Gamma\vdash N:\mathtt{unit}\to A
        \justifies
        \Gamma\vdash \mathtt{mkvar}~M~N: \mathtt{var}[A]
\endprooftree
\]

This language is equipped with a standard big-step call-by-value operational semantics. To define it, we temporarily extend the syntax of terms with identifiers
for \textbf{locations}, denoted by $l$. Then, \textbf{values} are formed as follows:

\[
V ::= \mathtt{skip}~|~\mathtt{true}~|~\mathtt{false}~|~\lambda x.M~|~\lpair{V, V}~|~l~|~\mathtt{mkvar}~V~V
\]

The operational semantics of $\Lang_2$ are then given as an inductively generated relation $(L, s) M \eval (L', s') V$, where
$s$ is a partial map from locations in $L$ to values. The rules for $\lambda$-calculus, products and booleans are standard (they do not affect the store) and we give in Figure \ref{opsem}
the rules for references. Note that as usual, some store annotations are omitted to aid readability;
the rules can be disambiguated as explained in \cite{ahm}. For a closed term $M$ without free locations, we write
$M\eval$ to indicate that $(\emptyset, \emptyset) M \eval (L, s) V$ for some $L, s$ and $V$. The observational preorder $M \obleq N$ between terms $M$ and $N$ is then defined as usual, by requiring that 
for all contexts $C[-]$ such that $C[M]$ and $C[N]$ are closed and contain no free location, if $C[M]\eval$ then $C[N]\eval$. The corresponding equivalence relation is denoted by $\obs$.

\begin{figure*}
{\footnotesize
\[
\prooftree
        M_1 \eval V_1~~~~M_2\eval V_2
        \justifies
        \mathtt{mkvar}~M_1~M_2\eval \mathtt{mkvar}~V_1~V_2
\endprooftree
~~~~~~~~
\prooftree
        \justifies
        (L, s)~\mathtt{new}_A \eval (L\cup \{l:A\}, s)~l
        \using (l\not \in L)
\endprooftree
~~~~~~~~
\prooftree
        (L,s)M\eval (L', s')l~~~~(L', s')N\eval (L'', s'')V
        \justifies
        (L, s)M:= N\eval (L'', s''(l\mapsto V))\mathtt{skip}
\endprooftree
\]
\vspace{10pt}
\[
\prooftree
        M\eval \mathtt{mkvar}~V_1~V_2
        ~~~
        N\eval V
        ~~
        V_1(V) \eval \mathtt{skip}
        \justifies
        M:= N \eval \mathtt{skip}
\endprooftree
~~~~~~~~
\prooftree
        (L, s)M\eval (L', s') l
        ~~~~
        s'(l) = V
        \justifies
        (L, s) !M \eval (L', s') V
\endprooftree
~~~~~~~~
\prooftree
        M\eval \mathtt{mkvar}~V_1~V_2
        ~~~~
        V_2(\mathtt{skip})\eval V
        \justifies
        !M \eval V
\endprooftree
\]
}
\caption{Big-step operational semantics for references in $\Lang_2$.}
\label{opsem}
\end{figure*}

\paragraph{Syntactic extensions} In this core language, one can define all the constructs of a basic imperative programming language. For instance if $C_1$ has type
$\mathtt{unit}$, sequential composition $C_1; C_2$ is given by:
\[
(\lambda d:\mathtt{unit}.~C_2)~C_1
\]
This works only because the evaluation of $\Lang_2$ is call-by-value. Likewise, a variable declaration $\mathtt{new}~x:A~\mathtt{in}~N$
(where $M$ has type $A$) can be obtained by
\[
(\lambda x:\mathtt{var}[A].~N)~\mathtt{new}_A
\]
and its initialized variant $\mathtt{new}~x=M~\mathtt{in}~N$ as expected.
As usual with general references one can define a fixed point combinator $Y$ by
\[
\begin{array}{l}
\lambda f:(A\to B)\to (A\to B).\\
~\mathtt{new}~y:A\to B~\mathtt{in}\\
~~y:=\lambda a:A.~f~!y~a;\\
~~!y
\end{array}
\]
This can be easily applied to implement a $\mathtt{while}$ loop. We can also use it to build an inhabitant $\bot$ to any type $A$.

\paragraph{Bad variables and isomorphisms} The $\mathtt{mkvar}$ construct allows to combine arbitrary ``write" and ``read" methods, forming terms of type $\mathtt{var}[A]$ not behaving
as reference cells: those are called \emph{bad variables}. We include bad variable in $\Lang_2$ for two reasons. First, because games models that allow bad variables are notably
simpler than those which do not \cite{DBLP:conf/fossacs/MurawskiT09}, for which it is not clear whether our methods apply. Second, because we expect the problem of isomorphisms without
bad variables to be far more subtle than what we consider here, because of the observation by O'Hearn that without bad variables, not only $\mathtt{var}[X]$ is not functorial, but it does
not even preserve isomorphisms. However, note that $\texttt{var}$-free isomorphisms are the same with or without bad variables.

\paragraph{Isomorphisms of types} We are now ready to define the notion of isomorphism of types in $\Lang_2$.
\begin{definition}
If $A$ and $B$ are two types of $\Lang_2$, we say that $A$ and $B$ are \emph{isomorphic}, denoted by $A\iso_{\Lang_2} B$,  if and only if there are two terms $x:A \vdash M:B$ and $y:B \vdash N:A$ such that:
\begin{eqnarray*}
(x:A \vdash (\lambda y.N) M) &\obs& id_A\\
(y:B \vdash (\lambda x.M) N) &\obs& id_B
\end{eqnarray*}
where $id_A = x:A \vdash x:A$.
\end{definition}


\subsection{The games model}
\label{section_games}
We now describe the fully abstract games model of $\Lang_2$. Note that except a few details there is nothing new here, as this is precisely the model described in \cite{ahm}. We however include the definitions
(but no proofs) for the sake of self-completeness.
\paragraph{Arenas, plays}
Our games have two participants: Player (P) and Opponent (O). Valid plays between $O$ and $P$ are generated by directed graphs called \emph{arenas}, which are the abstract representation of types. An \textbf{arena} is
a tuple $A = \lpair{M_A, \lambda_A, I_A, \enb{A}}$ where
\begin{itemize}
\item $M_A$ is a set of \textbf{moves},
\item $\lambda_A: M_A \to \{O, P\}\times \{Q, A\}$ is a \textbf{labeling} function which indicates whether a move is by Opponent or Player, and whether it is a Question or Answer. We write
\begin{eqnarray*}
\{O, P\}\times \{Q, A\} &=& \{OQ, OA, PQ, PA\}\\
\lambda_A &=& \lpair{\lambda^{OP}_A, \lambda^{AQ}_A}
\end{eqnarray*}
The function $\overline{\lambda_A}$ denotes $\lambda_A$ with the $O/P$ part reversed. A move $a\in M_A$ is a $O$-move (resp. $P$-move) if $\lambda_A(a) = O$ (resp. $\lambda_A(a) = P$).
\item $I_A\subseteq {\lambda_A}^{-1}(\{OQ\})$ is a set of \textbf{initial moves}
\item $\enb{A}\subseteq M_A^2$ is a relation called \textbf{enabling}, which satisfies that if $a \enb{A} b$, then $\lambda_A^{OP}(a) \neq \lambda_A^{OP}(b)$, and if $\lambda_A^{QA}(b) = A$ then
$\lambda_A^{QA}(a) = Q$.
\end{itemize}
We require two additional conditions on arenas: they should be \textbf{complete} (for each question $m\in M_A$, there should be an answer $n\in M_A$ such that $m \enb{A} n$) and \textbf{finitely branching} (for all $a\in M_A$, the set $\{m\in M_A\mid a\vdash_A m\}$ is finite).
We consider the usual arrow construction $A\tto B$ on arenas, as well as products $\Pi_{i\in I} A_i$ and lifted sums $\Sigma_{i\in I} A_i$ of finite families
of arenas. Their definitions can be found, for example, in \cite{ahm}. It is obvious
that they preserve the fact of being complete and finitely branching. The $0$-ary product (the empty arena) is denoted by $\one$, and will be a terminal object in
our category.

If $A$ is an arena, a \textbf{justified sequence} over $A$ is a sequence of moves in $M_A$ together with \textbf{justification pointers}: for each non-initial move $b$, there is a pointer to
an earlier move $a$ such that $a\enb{A} b$. In this case, we say that $a$ \textbf{justifies} $b$. The transitive closure of the justification relation is called \textbf{hereditary justification}.

\paragraph{Notations} The relation $\sqsubseteq$ will denote the prefix ordering on justified sequences. By $s \sqsubseteq^P t$, we mean that $s$ is a $P$-ending prefix of $t$.
If $s$ is a sequence, then $|s|$ will denote its length. We also define the \textbf{prefix functions} $\ip$ and $\jp$ by $\ip(\epsilon) = \epsilon$ and $\ip(sa) = s$, and
$\jp(si) = \epsilon$ if $i$ is initial, $\jp(s_1 a s_2 b) = s_1 a$ if $b$ is justified by $a$.

A justified sequence $s$ over $A$ is a \textbf{legal play} if it is:
\begin{itemize}
\item \textbf{Alternating}: If $s'ab \sqsubseteq s$, then $\lambda_A^{OP}(a) \neq \lambda_A^{OP}(b)$.
\item \textbf{Well-bracketed}: a question $q$ is \textbf{answered} by a later answer $a$ if $q$ justifies $a$. A justified sequence $s$ is well-bracketed if each answer is justified by the last
unanswered question, that is, the \textbf{pending} question.
\end{itemize}
The set of all legal plays on $A$ is denoted by $\leg{A}$. We will also be interested in the set $\preleg{A}$ of well-bracketed but not necessarily alternating plays on $A$, called
\textbf{pre-legal plays}. 

\paragraph{Strategies, composition}

A \textbf{strategy} $\sigma$ on an arena $A$ (denoted $\sigma: A$) is a non-empty set of $P$-ending legal plays on $A$ satisfying \textbf{prefix-closure}, \emph{i.e.} that for all $sab \in \sigma$,
we have $s\in \sigma$ and \textbf{determinism}, \emph{i.e.} that if $sab, sac\in \sigma$, then $b=c$. As usual, strategies form a category which has
arenas as objects, and strategies $\sigma: A\tto B$ as morphisms from $A$ to $B$. If $\sigma : A\tto B$ and $\tau: B \tto C$ are strategies, their composition $\sigma; \tau: A\tto C$
is defined as usual by first defining the set of \textbf{interactions} $u\in I(A, B, C)$ of plays $u\in \leg{(A\tto B)\tto C}$ such that $u_{\restrict A, B}\in \leg{A\tto B}$,
$u_{\restrict B, C} \in \leg{B\tto C}$ and $u_{\restrict A, C}\in \leg{A\tto C}$ (where $s_{\restrict A, B}$ is the usual restriction operation essentially taking the subsequence of $s$ in $M_A$ and $M_B$,
along with the possible natural reassignment of justification pointers). The \textbf{parallel interaction} of $\sigma$ and $\tau$ is then the set 
$\sigma||\tau = \{u\in I(A, B, C) \mid u_{\restrict A, B} \in \sigma \wedge u_{\restrict B, C} \in \tau\}$, and the composition of $\sigma$ and $\tau$ is obtained by
the \textbf{hiding} operation, \emph{i.e.} $\sigma; \tau = \{u_{\restrict A, C}\mid u\in \sigma||\tau\}$.
It is known (e.g. \cite{McCusker1996}) that composition is associative. It admits \emph{copycat strategies} as identities:
$id_A = \{s\in \leg{A_1\tto A_2}\mid \forall s'\sqsubseteq^P s, s'_{\restrict A_1} = s'_{\restrict A_2}\}$.

If $s\in \leg{A}$, the \textbf{current thread} of $s$, denoted $\thread{s}$, is the subsequence of $s$ consisting of all moves hereditarily justified by the same initial move as the last
move of $s$. All strategies we are interested in will be \emph{single-threaded}, \emph{i.e.} they only depend on the current thread. Formally, $\sigma: A$ is \textbf{single-threaded} if
\begin{itemize}
\item For all $sab\in \sigma$, $b$ points in $\thread{sa}$,
\item For all $sab, t\in \sigma$ such that $ta\in \leg{A}$ and $\thread{sa} = \thread{ta}$, we have $tab\in \sigma$.
\end{itemize}
It is straightforward to prove that single-threaded strategies are stable under composition and that $id_A$ is single-threaded. Hence, there is a category $\Gam$ of arenas and single-threaded strategies.
The category $\Gam$ will be the base setting for our analysis. Given arenas $A$ and $B$, the arena $A\times B$ defines a cartesian product of $A$ and $B$ and the construction $A\tto B$ extends
to a right adjoint $A\times - \dashv A\tto -$, hence $\Gam$ is cartesian closed and is a model of simply typed $\lambda$-calculus. It also has \emph{weak coproducts}, given by the lifted sum \cite{ahm}.

\paragraph{Views, classes of strategies}

In this paper, we are mainly interested in the properties of single-threaded strategies. However, to give a complete account of the context it seems necessary to mention several classes of strategies
of interest in this setting. The most important one is certainly the class of \emph{innocent} strategies, both for historical reasons and because it is at the core of the frequent definability results -- and thus of
the full abstraction results -- in game semantics. Their definition requires the notion of $P$-view, defined as usual by induction on plays as follows.
\[
\begin{array}{rcll}
\pview{si} &=& i & \text{if $i\in I_A$}\\
\pview{sa} &=& \pview{s} &\text{if $\lambda_A^{OP}(a) = P$}\\
\pview{s_1 a s_2 b} &=& \pview{s_1} a b& \text{if $\lambda_A^{OP}(b) = O$ and $a$ justifies $b$}
\end{array}
\]
A strategy $\sigma: A$ is then said to be \textbf{visible} if it always points inside its $P$-view, that is, for all $sab\in \sigma$ the justifier of $b$ appears in $\pview{sa}$. The strategy $\sigma$ is 
\textbf{innocent} if it is visible, and if its behaviour only depends on the information contained in its $P$-view. More formally, whenever $sab, t\in \sigma$ such that $ta\in \leg{A}$ and $\pview{sa} = \pview{ta}$,
we must also have $tab\in \sigma$. Both visibility and innocence are stable under composition \cite{hyland-ong,abramsky-mccusker:active-algol}, thus let us denote by $\Vis$ the category of
arenas and visible single-threaded strategies and by
$\Inn$ the category of arenas and innocent strategies. Both categories inherit the cartesian closed structure of $\Gam$, but strategies in $\Inn$ are actually nothing but abstract representation of ($\eta$-long $\beta$-normal)
$\lambda$-terms and form a fully complete model of simply-typed $\lambda$-calculus. Strategies in $\Vis$ have more freedom, they correspond in fact to programs with first-order store \cite{abramsky-mccusker:active-algol}.

\subsection{Modeling $\Lang_2$ in $\Fam(\Gam)$}

\paragraph{Interpretation}
The three categories $\Gam$, $\Vis$ and $\Inn$ are categories of \emph{negative games} (in which Opponent always plays first), and these are known to model call-by-name computation
whereas $\Lang_2$ is call-by-value.
We could have modeled it using positive games, following the lines of \cite{DBLP:journals/tcs/HondaY99}. Instead, we follow \cite{ahm}
and model $\Lang_2$ in the free completion $\BFam(\Gam)$ of $\Gam$ with respect to coproducts. This will allow us to first characterize
isomorphisms in $\Gam$ (result which could be applied to a call-by-name language with state) then deduce from it the isomorphisms in $\BFam(\Gam)$.
In fact we will consider the completion $\Fam(\Gam)$ of $\Gam$ with respect to \emph{finite} coproducts, since $\Lang_2$ has only finite types.

The objects of $\Fam(\Gam)$ are finite families $\{A_i\mid i\in I\}$ of arenas. A map from $\{A_i\mid i \in I\}$ to $\{B_j\mid j\in J\}$ is the data of a function $f: I\to J$ together with a family of strategies
$\{\sigma_i:A_i\to B_{f(i)}\mid i\in I\}$. As shown in \cite{abramsky-mccusker:families}, the cartesian closed structure of $\Gam$ extends to $\Fam(\Gam)$. Moreover, the weak coproducts in $\Gam$ give rise to a \emph{strong monad} $T$ on $\Fam(\Gam)$.
Given families $A = \{A_i\mid i \in I\}$ and $B = \{B_j\mid j\in J\}$, we define
\begin{eqnarray*}
A\times B &=& \{A_i \times B_j\mid (i,j) \in I\times J\}\\
A\tto B &=& \{\Pi_{i\in I} (A_i\tto B_{f(i)})\mid f: I \to J\}\\
TA&=& \{\Sigma_{i\in I} A_i\}
\end{eqnarray*}
The singleton family $\{\one\}$ is the terminal object of $\Fam(\Gam)$. By abuse of notation, we will still denote it by $\one$.
The other components of the cartesian closed structure of $\Fam(\Gam)$ and of the strong monad structure of $T$ follow naturally from these definitions. We skip the details, as all of this is already covered in \cite{ahm}. It is
known that given a cartesian closed category with a strong monad, one can interpret call-by-value languages in the Kleisli category of the monad \cite{moggi}, and the interpretation of $\Lang_2$ in $\Fam(\Gam)$ follows these lines.

We interpret $\mathtt{unit}$ and $\mathtt{bool}$ as the families with respectively one and two elements whose components are all empty arenas.
Of course we also need to give an interpretation for $\var{A}$, along with morphisms for the read and write operations of the reference cell. Once again, we follow the lines of \cite{ahm} and consider the type $\var{A}$
as the product of its read and write methods, hence we set $\intr{\var{A}} = (\intr{A} \tto T\one) \times T\intr{A}$. The interpretation relies on the definition of a morphism $\one \to \intr{\var{A}}$,
that is, if $\intr{A} = \{A_i\mid i \in I\}$, a strategy $\mathtt{cell}: (\Pi_{i\in I} (A_i\tto \one_\bot)\times \Sigma_{i\in I} A_i)_{\bot}$, where $A_\bot = T\{A\}$ is the lift operation. Apart from the initial
protocol due to the lift, the strategy $\mathtt{cell}$ works by associating each read request with the latest write request and playing copycat between them. A more detailed description is given in \cite{ahm}, and an
algebraic definition is obtained in \cite{DBLP:journals/entcs/MelliesT09}. As proved in \cite{ahm}, this gives a sound interpretation of $\Lang_2$ in $\Fam(\Gam)_T$.

\paragraph{Complete plays and full abstraction}
A fully abstract model of $\Lang_2$ is obtained by quotienting $\Fam(\Gam)_T$ by the usual observational preorder. However, as is often the case with fully abstract game semantics of languages with 
store, it is \emph{effectively presentable}: the observational preorder can be characterized directly. Say a play $s\in \leg{A}$ is \textbf{complete} if
it has as many questions and answers, \emph{i.e.} all questions have been answered. If $\sigma$ is a strategy on an arena $A$, then let us denote by $\comp(\sigma)$ the set of complete plays
in $\sigma$. Take $\sigma, \tau: A\to B$ two morphisms in $\Fam(\Gam)_T$, with $A = \{A_i \mid i\in I\}$ and $B = \{B_j\mid j\in J\}$. Then, $\sigma$ and $\tau$ consist of families
$\{\sigma_i : A_i \to \Sigma_{j\in J} B_j\mid i\in I\}$ and $\{\tau_i : A_i \to \Sigma_{j\in J} B_j\mid i\in I\}$. We then say that $\sigma \sleq \tau$ if and only if
for all $i\in I$, $\comp(\sigma_i) \subseteq \comp(\tau_i)$.

Take now two terms $M$ and $N$ of type $A$, and suppose $\intr{A} = \{A_i \mid i\in I\}$. Then 
$\intr{M}$ and $\intr{N}$ are morphisms from $\one$ to $T \intr{A}$ in $\Fam(\Gam)$, \emph{i.e.} strategies on $\Sigma_{i\in I} A_i$. 
Given the full abstraction result of \cite{ahm}, it is then straightforward to prove the following equivalence:
\[
M \obleq N \Leftrightarrow \intr{M} \sleq \intr{N}
\]
This concrete representation of the observational preorder will be central to our characterization of isomorphisms in $\Lang_2$.

\section{Isomorphisms in $\Gam$}
\label{section_isomorphisms}
We are now going to extend Laurent's tools \cite{DBLP:journals/mscs/Laurent05} to characterize isomorphisms of types for $\Lang_2$. We will first recall Laurent's work in the visible and innocent
cases, then extend it to characterize isomorphisms in $\Gam$. From there, we will switch to call-by-value and study the isomorphisms in $\Fam(\Gam)$, and in its Kleisli category
over $T$. 

\subsection{Isomorphisms and zig-zag strategies}

We start by giving the definition of (a subtle adaptation of) Laurent's \emph{zig-zag} plays.

\begin{definition}
Let $s\in \leg{A\tto B}$ be a legal play. It is \textbf{zig-zag} if 
\begin{itemize}
\item[1.] Each $P$-move following an $O$-move in $A$ (resp. in $B$) is in $B$ (resp. in $A$),
\item[2.] A $P$-move in $A$ immediately follows an initial $O$-move in $B$ if and only if it is justified by it, 
\item[3.] The (not necessarily legal) sequences $s_{\restrict A}$ and $s_{\restrict B}$ have the same pointers.
\end{itemize}
If $s$ only satisfies the first two conditions, then it is \textbf{pre-zig-zag}.
\end{definition}

By extension, we will say that a strategy $\sigma$ is \textbf{pre-zig-zag} (resp. \textbf{zig-zag}) if all its plays are so. 
The core of Laurent's theorem is then that all isomorphisms in $\Vis$ are zig-zag strategies. His proof does rely on visibility, however
it only gets involved to prove that the condition $3$ of zig-zag plays
is satisfied. The first half of his argument does not use visibility and actually proves that all isomorphisms in $\Gam$ are pre-zig-zag.
Here, being mainly interested in $\Gam$, we make this explicit. We need first the following lemma.

\begin{lemma}[Dual pre-zig-zag play]
Let $s\in \leg{A\tto B}$ be a pre-zig-zag play, then there exists an unique pre-zig-zag $\overline{s}\in \leg{B\tto A}$ such that $\overline{s}_{\restrict A} = s_{\restrict A}$
and $\overline{s}_{\restrict B} = s_{\restrict B}$.
\end{lemma}
\begin{proof}
We define $\overline{s}$ by induction on $s$; $\overline{\epsilon} = \epsilon$, and $\overline{sab} = \overline{s}ba$. We keep the same pointers, except for the case where a move
$a$ in $A$ was justified by an initial move $b$ in $B$. Then because of the pre-zig-zag condition on $s$, $a$ is necessarily an initial move in $A$ and 
is set as the new justifier of $b$ in $\overline{s}$. There is no other possible $\overline{s}$, since the restrictions on $A$ and $B$ are constrained by the hypotheses
and their interleaving is forced by the alternation and the pre-zig-zag conditions on $\overline{s}$.
\end{proof}

\begin{lemma}
If $\sigma: A\tto B$, $\tau: B\tto A$ form an isomorphism in $\Gam$, then they are pre-zig-zag and
for all $s$, $s\in \sigma \Leftrightarrow \overline{s}\in \tau$.
\end{lemma}
\begin{proof}
Consider an isomorphism $\sigma: A\tto B$, $\tau: B\tto A$ in $\Gam$. We will prove by induction on even $k\in \mathbb{N}$ that all plays
of $\sigma, \tau$ whose length is less than $k$ are pre-zig-zag, and that moreover 
$\{\overline{s}\mid s \in \sigma \wedge |s|\leq k\} = \{s\in \tau\mid |s|\leq k\}$.

If $k=0$, this is trivial.
Otherwise, suppose this is true up to $k\in \mathbb{N}$, and consider $sab\in \sigma$ of length $k+2$; let us first prove condition (1).
Without loss of generality, suppose $a\in M_A$. Since $s_{\restrict B} = \overline{s}_{\restrict B}$, by
a straightforward zipping argument we can build an interaction $u\in I(A_1, B, A_2)$
such that $u_{\restrict A_1, B} = s$ and $u_{\restrict B, A_2} = \overline{s}$, moreover since $\sigma, \tau$ form an isomorphism we
must have $u_{\restrict A_1, A_2} \in id_A$. Now, we necessarily have $b\in M_B$, otherwise $u$ could be extended to $uab\in \sigma||\tau$
with $uab_{\restrict A_1, A_2} = u_{\restrict A_1, A_2} ab$ which is not a play of the identity, contradiction. Hence $sab$ satisfies condition $1$ of
pre-zig-zag plays. 

To see why it satisfies condition $2$, take $sba\in \sigma$ with $b$ in $B$ and $a$ in $A$. If $b$ is initial in $B$, then $a$ necessarily points to it
since $\sigma$ is single-threaded. Reciprocally, suppose $a$ points to an initial move in $B$ earlier than $b$. Then we have $\overline{s}\in \tau$, and by
the same zipping argument as above we have an unique $u\in I(B_1, A, B_2)$ such that $u_{\restrict B_1, A} = \overline{s}$ and $u_{\restrict A, B_2} = s$. Since
$\sigma, \tau$ form an isomorphism we also have $u_{\restrict B_1, B_2} \in id_B$. Let us now extend $u$ to $u' = u b_2 a b_1$ in the unique way such that 
$u'_{\restrict A, B_2} = sba$ and $u'_{\restrict B_1, A} \in \tau$. Note that we are sure that $b_1$ is a move on $B_1$ since $\overline{s} a b_1$ is a play
of $\tau$ of length $k+2$ and we already know that these satisfy the condition $1$ of pre-zig-zag plays. But we also have $u'_{\restrict B_1, B_2} \in id_B$,
hence $b_2$ points in $\overline{s}$ as $b_1$ points in $s$. This means that we have $\overline{s} a b \in \tau$, such that $a$ is initial and $b$ points in
$\overline{s}$, impossible since $\tau$ is single-threaded. Hence $sba$ satisfies condition $2$ of pre-zig-zag plays.

We have proved that $sab$ is pre-zig-zag, so $\overline{sab}$ is defined. By induction hypothesis $\overline{s}\in \tau$ and the same reasoning as above
shows that it extends to $\overline{sab}\in \tau$. The argument is symmetric, 
hence $\{\overline{s}\mid s \in \sigma \wedge |s|\leq k+2\} = \{s\in \tau\mid |s|\leq k+2\}$.
\end{proof}

For the sake of completeness, let us include Laurent's argument which proves that isomorphisms in $\Vis$ are zig-zag.

\begin{lemma}
If $\sigma: A\tto B$, $\tau: B\tto A$ form an isomorphism in $\Vis$, then $\sigma$ and $\tau$ are zig-zag strategies.
\end{lemma}
\begin{proof}
We already know that $\sigma$ and $\tau$ are pre-zig-zag strategies. We show by induction on $n\in \mathbb{N}$ that for all $s\in \sigma$, if $|s|\leq n$
then $s_{\restrict A}$ and $s_{\restrict B}$ have the same pointers. Take now $s\in \sigma$, and $sab \in \sigma$, suppose w.l.o.g. that $a\in M_A$. Suppose
$a$ points to $(s_{\restrict A})_i$, then $b$ points to $(s_{\restrict B})_i$. Indeed, it cannot point to $(s_{\restrict B})_j$ with $j>i$ since that would
break visibility for $\sigma$. But if it points to $(s_{\restrict B})_j$ with $j<i$ we use the same reasoning on the dual pre-zig-zag play $\overline{sab}$
and get a contradiction with the fact that $\tau$ is visible.
\end{proof}

Let us denote by $\Gam_i$, $\Vis_i$ and $\Inn_i$ the groupoids of arenas and isomorphisms on the respective categories.
In the next sections, we use these facts to give more combinatorial representations of $\Gam_i$, $\Vis_i$ and $\Inn_i$.

\subsection{Notions of game morphisms}

Laurent's isomorphism theorem works by relating isomorphisms in $\Gam$ with isomorphisms in a simpler category which has arenas as objects and \emph{forest 
morphisms}\footnote{Note that in \cite{DBLP:journals/mscs/Laurent05} arenas are forests, which is not the case here.}, \emph{i.e.} maps on moves that preserve initiality
and enabling. Relaxing the visibility conditions requires us to also consider relaxed notions of game morphisms, that we present here.

\begin{definition}
Let $A$ be an arena. A \textbf{path} on $A$ is a play $s\in \leg{A}$ such that except for the initial move, every move in $s$ points to the previous move.
Formally, for all $s'ab\sqsubseteq s$, $a$ justifies $b$ in $s$. Let $\path{A}$ denote the set of paths on $A$. A \textbf{path morphism} from $A$ to $B$
is a function $\phi:\path{A} \to \path{B}$ such that $\ip\circ \phi = \phi\circ \ip$ and which preserves $Q/A$ labeling: for all $sa\in \path{A}$ with
$\phi(sa) = \phi(s) b$, we have $\lambda_A^{QA}(a) = \lambda_B^{QA}(b)$. There is a category $\Path$ of arenas and path morphisms.
\end{definition}

This category $\Path$ comes with its own notion of isomorphisms of arenas. Note that whenever $A$ is a forest, this is exactly Laurent's notion
of forest isomorphism. We now introduce two weaker notions of morphisms for arenas. In what follows, let us call a legal play on $A$ with only
one initial move a \textbf{thread} on $A$, and denote the set of threads on $A$ by $\threads{A}$. Likewise, 
let us call a pre-legal play with one initial move a pre-legal thread and let us denote these by $\prethreads{A}$.

\begin{definition}
Let $A$, $B$ be arenas, and let $\phi : \prethreads{A} \to \prethreads{B}$
We say that $\phi$ is a \textbf{sequential morphism}
from $A$ to $B$ if $\ip\circ \phi = \phi\circ \ip$, and if it preserves $Q/A$ labeling, \emph{i.e.} for all $\phi(sa) = \phi(s) b$ we have $\lambda_A^{QA}(a) = \lambda_B^{QA}(b)$.
We say that it is a \textbf{justified morphism} if, additionally, $\jp\circ \phi = \phi \circ \jp$.
There are two categories $\Seq$ of arenas and sequential morphisms and $\Jus$ of arenas and
justified morphisms.
\end{definition}

As above, we will denote by $\Seq_i$, $\Jus_i$ and $\Path_i$ the groupoids of invertible maps in $\Seq$, $\Jus$ and $\Path$. These groupoids
will soon appear to be identical to $\Gam_i$, $\Vis_i$ and $\Inn_i$. To prove this, we need the following lemma.

\begin{lemma}
Let $s\in \prethreads{A}$, and $\sigma:A\tto B$ an isomorphism in $\Gam$. There is then an unique play $s'\in \sigma$ such that $s'_{\restrict A} = s$.
\end{lemma}
\begin{proof}
Remark first that if $\sigma: A\tto B$ and $\tau: B \tto A$ are inverses then they are both total, \emph{i.e.} for all $s\in \sigma$ and $sa\in \leg{A\tto B}$
there must be $b$ such that $sab\in \sigma$, assuming it is not the case easily leads to a contradiction. We now prove the lemma by induction on $s$. If $s = \epsilon$,
this is trivial. Otherwise, suppose $sa\in \prethreads{A}$ and we have by induction hypothesis $s'\in \sigma$ such that $s'_{\restrict A} = s$. If $a$ is a
$P$-move in $A$ (hence an $O$-move in $A\tto B$), there is an unique $b$ such that $s'ab\in \sigma$, and we do have $s'ab_{\restrict A} = sa$. If
$a$ is an $O$-move in $A$ (hence a $P$-move in $A\tto B$), then let $\tau: B\tto A$ be the inverse of $\sigma$, since $s'\in \sigma$ we have 
$\overline{s'}\in \tau$. Being part of an isomorphism $\tau$ is total, hence there is $b$
such that $\overline{s'}ab\in \tau$. We deduce from this that $s' b a\in \sigma$, and we have $s'ba_{\restrict A} = sa$ as needed. This choice is unique:
if there is another play $t\in \sigma$ such that $t_{\restrict A} = sa$, then $t = t'b'a$ (since $t$ is zig-zag). By induction
hypothesis we have $t' = s'$, thus $s'b'a\in \sigma$. From this we deduce that $\overline{s'}ab'\in \tau$, so $b=b'$ by determinism of $\tau$.
\end{proof}

\begin{proposition}
If $\C \iso \D$ means that two groupoids $\C$ and $\D$ are \emph{isomorphic}, then we have:
\begin{eqnarray*}
\Gam_i &\iso& \Seq_i\\
\Vis_i &\iso& \Jus_i
\end{eqnarray*}
\end{proposition}
\begin{proof}
Let us first define a functor $F: \Gam_i \to \Seq_i$. It is defined as the identity on arenas. Let $\sigma : A\tto B$ be an isomorphism, and
let $s\in \prethreads{A}$ then we define $\phi_\sigma(s) = s'_{\restrict B}$, where $s'$ is the unique play on $A\tto B$ which existence is ensured by the lemma above.
The function $\phi_\sigma$ commutes with $\ip$ since $\sigma$ is a pre-zig-zag strategy. To any question it cannot associate an answer, as that would immediately break well-bracketing on $\sigma$.
But to any answer it cannot associate a question, as that would immediately break well-bracketing on $\sigma^{-1}$.
Then we define $F(\sigma) = \phi_\sigma$. It is obvious that $F$ preserves identities and composition\footnote{In fact, this construction can be seen as a particular case
of Hyland and Schalk's faithful functor from games to relations \cite{hyland-schalk}, where the relation happens to be functional.}.

Reciprocally, suppose $\phi: A\to B$ is a sequential isomorphism. We mimic the usual definition of the identity by setting
$G(\phi) = \{s\in \leg{A\tto B}\mid \forall s'\sqsubseteq^P s,~ \phi(s'_{\restrict A}) = s'_{\restrict B}\}$ (We apply $\phi$ on
plays whereas it is normally only defined on \emph{threads}, however it can be canonically extended to plays, so this is not ambiguous).
It is obvious that this construction is functorial, and that it is inverse to $F$.

We have now an isomorphism $\Gam_i \iso \Seq_i$ which restricts naturally to $\Vis_i$ and $\Jus_i$. Indeed if $\sigma: A\tto B$ is a visible isomorphism,
it is a zig-zag strategy therefore $s\in \prethreads{A}$ and $\phi_\sigma(s)$ have the same pointers, which means that $\jp\circ \phi_\sigma = \phi_\sigma \circ \jp$.
Reciprocally if $\phi_\sigma$ is a justified morphism, all $s\in \sigma$ must be such that $s_{\restrict A}$ and $s_{\restrict B}$ have the same pointers,
therefore $\sigma$, being pre-zig-zag, always points in its $P$-view.
\end{proof}

\subsection{Innocent and visible case}

In this section, we use the framework described above to recall Laurent's results. We have proved above that isomorphisms in $\Vis$ correspond
to isomorphisms in $\Jus$, which we are now going to compare with isomorphisms in $\Path$. 

\begin{lemma}
There is a full functor $H: \Vis_i \to \Path_i$.
\end{lemma}
\begin{proof}
We have built in the above section a full and faithful functor (actually an isomorphism) $F: \Vis_i \to \Jus_i$. From a visible isomorphism $\sigma : A\tto B$
we set $H(\sigma) = F(\sigma) \restrict \path{A}$, where $f \restrict E'$ restricts a function $f: E\to F$ to a subset $E'\subseteq E$ of its domain.
The image of a path by $F(\sigma)$ is always a path since it is a justified morphism, hence $H(\sigma) : \path{A} \to \path{B}$.

To see why $H$ is full, suppose we have a path morphism $\phi: \path{A} \to \path{B}$. Then $\phi$ admits a canonical extension $\phi^*: \prethreads{A} \to \prethreads{B}$. To
define $\phi^*(s)$ we reason by induction on $s$, and set $\phi^*(\epsilon) = \epsilon$ and $\phi^*(sa) = \phi^*(s) a'$, where $a'$ is the last move of $\phi(p_a)$, $p_a$
being the path of $a$ in $s$. The move $a'$ keeps the same pointer as $a$. It is clear that this defines as needed a justified morphism $\phi^*$ such that $H(\phi^*) = \phi$.
\end{proof}

This ensures that arenas $A$ and $B$ are isomorphic in $\Vis$ if and only if they are isomorphic in $\Path$, \emph{i.e.} they are geometrically the same. 
Let us mention that as Laurent proved, this correspondence is one-to-one in the innocent case: one can prove that there is only one innocent zig-zag strategy corresponding
to a particular path isomorphism, hence $H$ restricts to an isomorphism of groupoids $H': \Inn_i \to \Path_i$.

\paragraph{Faithfulness of $H$}
Note however that $H$ itself is \emph{not} faithful, because we can exploit non-innocence to build non-uniform isomorphisms, \emph{i.e.} isomorphisms which change
their underlying path isomorphism as the interaction progresses. For an example, consider the arena $A = \raisebox{20pt}{
\xymatrix@R=2pt{
&&q\\
q_1\ar@{-}@/^/[urr]&q_2\ar@{-}@/^/[ur]&a\ar@{-}[u]\\
a\ar@{-}[u]&a\ar@{-}[u]
}}
$
which is the interpretation of $(\mathtt{bool}\to \mathtt{unit}) \to \mathtt{unit}$ in call-by-value and of $\mathtt{unit}\times \mathtt{unit} \to \mathtt{unit}$ in
call-by-name. Consider now the strategy $i:A\tto A$ which behaves as follows. It starts by playing as the identity on $A$. The first time Opponent
plays $q_1$ or $q_2$ on the left hand side, it simply copies it. Starting from the second time Opponent plays $q_1$ or $q_2$ though, it swaps them. An example
play of $i$ is given in Figure \ref{involution}. Although it is not the identity, $i$ is its own inverse. Its image by $H$ only takes into account the first behaviour
or $i$, thus is the same as for $id_A$: the identity path morphism on $A$. From this strategy we can extract the following term $f:B \vdash M:B$ of $\Lang_2$, where
$B = (\mathtt{bool}\to \mathtt{unit}) \to \mathtt{unit}$.
\[
\begin{array}{l}
\mathtt{new}~r:=\mathtt{true}~\mathtt{in}\\
\lambda g. f (\lambda b. \mathtt{if}~!r~\mathtt{then}~r:=\mathtt{false};g~b~\mathtt{else}~g~(\mathtt{not}~b))
\end{array}
\]

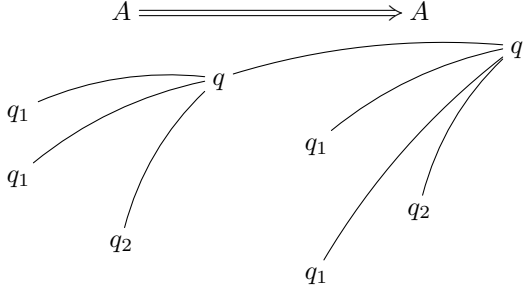
\begin{figure}
\[
\xymatrix@R=0pt{
			&A\ar@{=>}[rrr]		&			&			&A				&\\
			&			&			&			&				&q\\
			&			&q\ar@{-}@/^/[urrr]	&			&				&\\
q_1\ar@{-}@/^/[urr]	&			&			&			&				&\\
			&			&			&q_1\ar@{-}@/^/[uuurr]	&				&\\
q_1\ar@{-}@/^/[uuurr]	&			&			&			&				&\\
			&			&			&			&q_2\ar@{-}@/^/[uuuuur]		&\\
			&q_2\ar@{-}@/^/[uuuuur]	&			&			&				&\\
			&			&			&q_1\ar@{-}@/^/[uuuuuuurr]&				&
}
\]
\caption{A play of the non-trivial involution $i$ on $A$}
\label{involution}
\end{figure}
Although $M$ is not the identity it is an involution on $B$, \emph{i.e.} we have $(\lambda f. M) (M x) \obs_{\Lang_2} x$. Such non-trivial involutions
cannot be defined using only purely functional behaviour.

We give in Figure \ref{groupoids} a summary of all the groupoids of isomorphisms encountered for the moment, along with their relations. Following it,
the question of finding the isomorphisms in $\Gam$ boils down to the definition of an arrow from $\Seq_i$ to $\Path_i$ in this diagram, which is what we
will attempt in the next two subsections.

\begin{figure}
\[
\xymatrix@C=30pt@R=30pt{
\Inn_i  \ar[rr]^{\text{iso}}
        \ar[d]^{\subseteq}      &&\Path_i\\
\Vis_i  \ar[r]^{\text{iso}}
        \ar[d]^{\subseteq}      &\Jus_i \ar[ur]^{\text{full}}_{\text{(not faithful)}}\ar[dr]^{\subseteq}\\
\Gam_i  \ar[rr]^{\text{iso}}&&\Seq_i    \ar@{.>}[uu]_{?}
}
\]
\caption{Relations between all groupoids of isomorphisms}
\label{groupoids}
\end{figure}
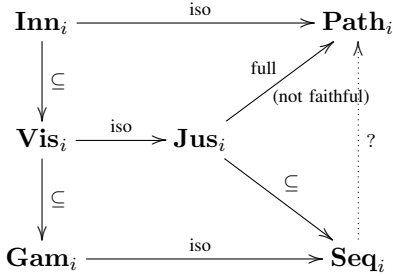

\subsection{Non-visible isomorphisms by counting}

We have seen above that we can build a full functor $\Vis_i \to \Path_i$, which allows to characterize isomorphic arenas in $\Vis$. However,
this construction relies heavily on visibility. We now investigate how to get rid of it and prove that two arenas $A$ and $B$ are isomorphic in $\Gam$ if and only if they are isomorphic in $\Path$.
In this subsection, we will describe for pedagogical reasons an intuitive approach to the proof, which relies on counting. However this approach
suffers from some defects, hence the full proof (described in the next subsection) will follow slightly different lines.

If $a\in M_A$, let us call its \textbf{arity}
the quantity $ar(a) = |\{m\in M_A\mid a \enb{A} m\}|$. On pre-threads $s\in \prethreads{A}$ we define:
\[
Q(s) = \sum_{i=1}^{|s|}{ar(s_i)}
\]
If $s\in \prethreads{A}$, $Q(s)$ is also the number of ways $s$ can be extended to some $sa$ (let us recall here that as a member of $\prethreads{A}$, $s$ need not be alternating): the choice of a justifier $s_i$
plus a move enabled by $s_i$. These definitions allow to express the following observation. If $\sigma: A\tto B$ is an isomorphism (thus a pre-zig-zag strategy) and $s\in \sigma$, then
$Q(s_{\restrict A}) = Q(s_{\restrict B})$, because $\sigma$ being an isomorphism, it must associate each possible extension of $s_{\restrict A}$ to an unique extension of $s_{\restrict B}$.
But this also means that if $sab\in \sigma$ we have $Q(s_{\restrict A}) + ar(a) = Q(s_{\restrict A} a) = Q(s_{\restrict B} b) = Q(s_{\restrict B}) + ar(b)$, hence $ar(a) = ar(b)$. Thus to
each move $a$, $\sigma$ must associate a move with the same arity. This is a step in the right direction, but we would like a deeper connection between $a$ and $b$.

If $a\in M_A$, we will use the notation $J_a = \{m\in M_A \mid a\vdash_A m\}$.
Let us define by induction on $k$ the notion of a $k$-isomorphism between $a\in M_A$ and $b\in M_B$. 
For any $a\in M_A$ and $b\in M_B$ there is automatically a $0$-isomorphism $i_{a, b}$. A $(k+1)$-isomorphism from $a$ to $b$
is the data of an isomorphism $f : J_a \to J_b$ along with, for all $m\in J_a$, a $k$-isomorphism $f_m: m \to f(m)$.
We use the notation $m \iso_k n$ to denote the fact that there is a $k$-isomorphism from $m$ to $n$. In other words, 
we have $m\iso_k n$ if the tree of paths of length at most $k$ starting form $m$ is tree-isomorphic to the tree of paths of length at most $k$ starting from $n$. 
If $k_1 \leq k_2$, $f_1$ is a $k_1$-isomorphism and $f_2$ is a $k_2$-isomorphism, we say that $f_1$ is a prefix of $f_2$ if they agree up to depth $k_1$.
Note that in particular we have $m \iso_1 n$ if and only if $ar(m) = ar(n)$, so $m \iso_k n$ is indeed a generalization of $ar(m) = ar(n)$. By induction on $k$, one can then prove that $\sigma$ must
always associate to each move $m$ a move $n$ such that $m \iso_k n$ : to prove it for $k+1$, just apply the counting argument above on $\iso_k$-equivalence classes. From all these $k$-isomorphisms,
one can then deduce the existence of a path isomorphism between $A$ and $B$. 

This counting argument has several unsatisfying aspects, which are caused by the implicit use of the following lemma.

\begin{lemma}[Slicing of isomorphisms]
Suppose $E'\subseteq E$ and $F'\subseteq F$ are finite sets, and that $f:E\to F$ and $g: E'\to F'$ are isomorphisms. Then there
is an isomorphism $f\setminus g: E\setminus E' \to F\setminus F'$.
\label{slicing_isos}
\end{lemma}

The obvious proof of this lemma is by cardinality reasons.
However this proof is, computationally speaking, ``almost non-effective", in the sense that the isomorphism it produces implicitly depends on the choice of a total
ordering for $E$ and $F$. A consequence of that is that from any isomorphism in $\Gam$ we will extract an isomorphism in $\Path$, but we cannot hope
its choice to be canonical, for any reasonable meaning of ``canonical". Even worse, the witness isomorphisms given by this proof for $\iso_k$ and $\iso_{k+1}$ need not agree together.
This implies that for infinitely deep arenas, one requires König's lemma to actually build a path isomorphism from a game isomorphism. This means that we cannot deduce from the proof above
an algorithm to extract path isomorphisms.

\subsection{Extraction of a path isomorphism}

To obtain a more computationally meaningful extraction of a path iso from a game iso, we must replace the proof of Lemma \ref{slicing_isos} by something
else than counting. As formalized in the following proof, the idea is to remark that given the data of Lemma \ref{slicing_isos}, starting from $x\in E\setminus E'$, the sequence
\begin{eqnarray*}
x_0 &=& f(x)\\
x_{n+1} &=& f\circ g^{-1}(x_n)
\end{eqnarray*}
must eventually reach $F\setminus F'$, as illustrated in Figure \ref{fig_slicing}, yielding a bijection between $E\setminus E'$ and $F\setminus F'$.

\begin{figure}
\begin{center}
\includegraphics[scale=0.5]{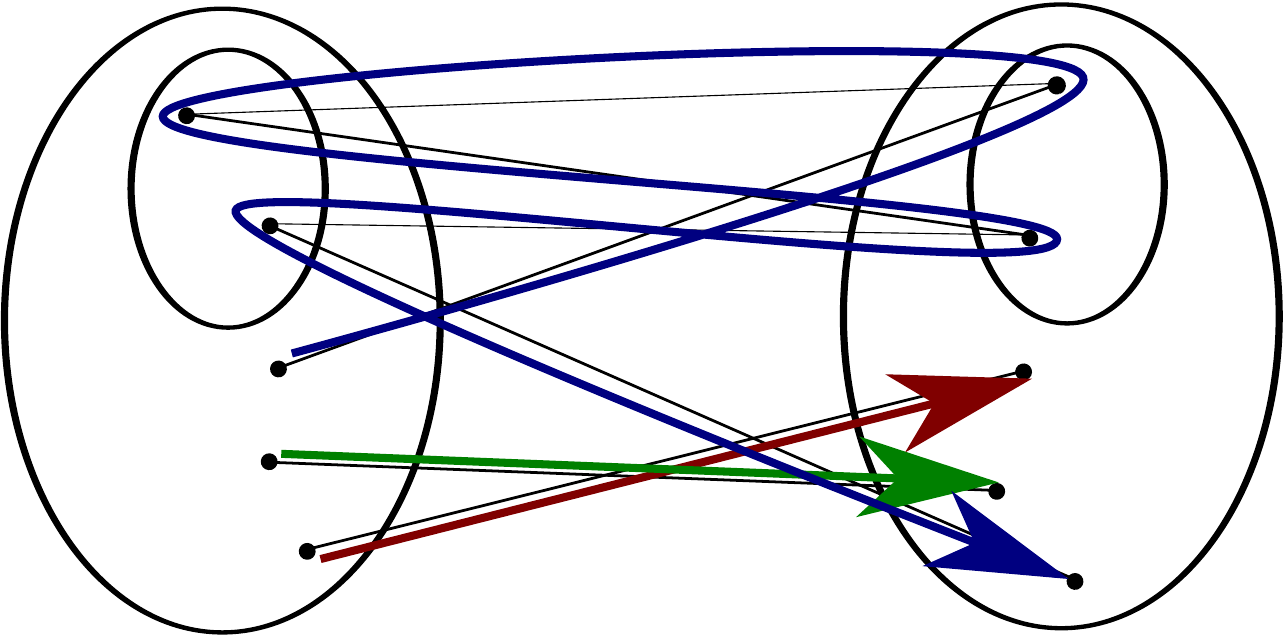}
\end{center}
\caption{Slicing of isomorphisms.}
\label{fig_slicing}
\end{figure}

\begin{proposition}
If $\phi : A\to B$ is a sequential play isomorphism, then for all $sa \in \prethreads{A}$ with $\phi(sa) = \phi(s) b$, there is a family
$(h_{s, sa}^k)_{k\in \mathbb{N}}$ such that for all $k$, $h_{s, sa}^k$ is a $k$-isomorphism from $a$ to $b$. 
This family is \emph{coherent}, in the following sense: if $k_1 \leq k_2$, $h_{s,sa}^{k_1}$ is a prefix of $h_{s, sa}^{k_2}$.
\end{proposition}
\begin{proof}
We will use the following notations. If $s\in \prethreads{A}$, $E_s$ will be the set of atomic extensions of $s$, that is of plays $sa\in \prethreads{A}$, and $F_s$ will be the set of
atomic extensions of $\phi(s)$. For all plays $sa\in \prethreads{A}$, although strictly speaking $E_s$ is \emph{not} a subset of $E_{sa}$, 
we have the following decomposition:
\[
E_{sa} = E_s + J_a
\]
Indeed, a move extending $sa$ can either point to some $s_i$ or to $a$.
Note also that for any $s$, $\phi: sa \mapsto \phi(s) b$ induces an isomorphism $f_s: a \mapsto b$ from $E_s$ to $F_s$.

For all $s\in \prethreads{A}$ and $sa\in E_s$, we follow the reasoning illustrated in Figure \ref{fig_slicing} and consider a bipartite directed graph $G_{s, sa}$ defined as follows: its set of vertices is $V = E_{sa} + F_{sa}$ and its set of edges is
$E = \{(x, f_{sa}(x)) \mid x\in E_{sa}\} + \{(y, f_s^{-1}(y))\mid y \in F_s\}$. This graph is ``deterministic", in the sense that the outwards degree of each vertex is at most one, moreover
the only vertices whose outwards degree is $0$ are those of $J_b$ (where $b = f_s(a)$, so $F_{sa} = F_s + J_b$). Moreover $G_{s, sa}$ must be acyclic, since $f_s$ and $f_{sa}$ are isomorphisms.
Thus from any vertex in $J_a$, there is an unique path in $G$ leading to a vertex in $J_b$; this induces an isomorphism $g_{s, sa} : J_a \to J_b$. For each pair $(m, g_{s, sa}(m))$ we
also keep track of the corresponding path $p_{s, sa}^m = (m, f_{sa}(m), f_s^{-1}(f_{sa}(m)), \dots, g_{s, sa}(m))$.

It is now time to build the $k$-isomorphisms, by induction on $k$. For $k=0$ this is obvious. For fixed $k+1\geq 1$, by induction hypothesis there is for each $sa\in \prethreads{A}$ with
$\phi(sa) = \phi(s) b$ a $k$-isomorphism $h_{s, sa}^k$ from $a$ to $b$. In particular, for fixed $sa\in \prethreads{A}$, consider the graph $G_{s, sa}$. 
Each of its edges of the form $(x, f_{sa}(x))$ are now labeled by the $k$-isomorphism $h_{sa, x}^k$ and all its edges of the form $(y, f_s^{-1}(y))$ are labeled by $(h_{s, f_s^{-1}(y)}^k)^{-1}$.
For each pair $(m, g_{s, sa}(m))$ we can now compose the labels along the path $p_{s, sa}^m$ and get a $k$-isomorphism $i_m : m \to g_{s, sa}(m)$. We then define 
$h_{s, sa}^{k+1} = (g_{s, sa}, (i_m)_{m\in J_a})$ which is as needed a $(k+1)$-isomorphism from $a$ to $b$.

Note finally that if $k_1 \leq k_2$, $h_{s, sa}^{k_1}$ is a prefix of $h_{s, sa}^{k_2}$. This is proved by simultaneous induction on $k_1$ and $k_2$. If $k_1 = 0$ this is obvious. Otherwise,
it relies on the fact that the graph $G_{s, sa}$ does not depend on $k$. Hence $h_{s, sa}^{k_1 + 1} = ( g_{s, sa}, (i_m)_{m\in J_a})$ and $h_{s, sa}^{k_2 + 1} = ( g_{s, sa}, (j_m)_{m\in J_a})$,
and each $i_m$ has be obtained from $k_1$-isomorphisms in the same way as $j_m$ has been obtained from $k_2$-isomorphisms, so it immediately boils down to the induction hypothesis.
\end{proof}

\begin{theorem}
Two finitely branching arenas $A$ and $B$ are $\Gam$-isomorphic if and only if they are $\Path$-isomorphic.
\label{main}
\end{theorem}
\begin{proof}
Consider an isomorphism $\sigma : A\tto B$ in $\Gam$. Restricted on plays with only two moves, it gives an isomorphism $f : I_A \to I_B$. By the previous proposition, there is for each $i\in I_A$
and for each $k\in \mathbb{N}$ a $k$-isomorphism $h^k_{\epsilon, i}: i \to f(i)$. Additionally, all these $k$-isomorphisms are compatible with each other, so they converge to an $\omega$-isomorphism
$h_{\epsilon, i}: i \to f(i)$. The iso $f$ together with $h_{\epsilon, i}$ for all $i$ define a path isomorphism from $A$ to $B$.
\end{proof}

\paragraph{Canonicity} For each pair of arenas $A, B$, we have built a function $K_{A, B} : \Gam_i(A, B) \to \Path_i(A, B)$. The natural question is then whether,
like in the other cases, this function defines a full functor. Unfortunately the answer is no, in fact $K_{A, B}$ is not even functorial. Indeed, the construction is based
on the more explicit proof of Lemma \ref{slicing_isos} illustrated in Figure \ref{fig_slicing}, which is not functorial; one can find
sets $E'\subseteq E$, $F'\subseteq F$ and $G' \subseteq G$ and isomorphisms $\xymatrix{E\ar[r]^{f} &F\ar[r]^{g}& G}$ and $\xymatrix{E'\ar[r]^{f'} &F'\ar[r]^{g'}& G'}$
such that $(f\setminus f'); (g\setminus g') \neq (f; g)\setminus (f'; g')$. From this it is not hard to find a counter-example to the functoriality
of $K_{A, B}$. It is a bit lengthy to describe it properly though, so we do not include it.

Although not being a functor, $K$ does satisfy some canonicity property: its result is invariant under renaming of moves in $A$ and $B$. In other terms,
$K_{A, B}: \Gam_i(A, B) \to \Path_i(A, B)$ is natural in $A$ and $B$, if both $\Gam_i(-, -)$ and $\Path_i(-,-)$ are seen as bifunctors
from $\Path_i^{op}\times \Path_i$ to $\Set$ (using implicitly the faithful functor from $\Path_i$ to $\Gam_i$ of Figure \ref{groupoids}).

\subsection{Application to $\Lang_2$}

Our isomorphism theorem most naturally applies to $\Gam$ (so to call-by-name languages), but $\Lang_2$ is modeled in $\Fam(\Gam)$, and more precisely in the Kleisli category of the strong monad $T$,
so we have to check how our result extends to this. Let us first relate isomorphisms in $\Fam(\Gam)_T$ and isomorphisms in $\Gam$ using the following lemma.

\begin{lemma}
Let $A = (A_i)_{i\in I}$ and $B = (B_j)_{j\in J}$, then isomorphisms between $A$ and $B$ in $\Fam(\Gam)_T$ are in one-to-one correspondence with pairs $(f, (\sigma_i)_{i\in I})$ where
$f: I \to J$ is a bijection and each $\sigma_i: A_i \tto B_{f(i)}$ is an iso in $\Gam$.
\label{analysis_isos}
\end{lemma}
\begin{proof}
Imagine $\sigma: A \to T(B)$ and $\tau:B \to T(A)$ are morphisms in $\Fam(\Gam)$ that form an isomorphism in $\Fam(\Gam)_T$. Here $A$ and $B$
are families of arenas, so $A = (A_i)_{i\in I}$, and $\sigma = (\sigma_i)_{i\in I}$ with $\sigma_i: A_i \tto T(B)$. Similarly, we have $B = (B_j)_{j\in J}$ and $\tau_j : B_j \tto T(A)$.
Then, first note that $\sigma_i$ and $\tau_j$ necessarily first give an answer to the initial Opponent move in $T$, \emph{i.e.} the initial question of the lifted sum in $T(B)$ and $T(A)$. Indeed take $i\in I$, and consider $\sigma_i; \tau^*: A_i \to T(A)$, where $\tau^*: T(B) \to T(A)$ is the \emph{lifting} of $\tau: B\to T(A)$.
This morphism must be a component of the identity on $A$ in $\Fam(\Gam)_T$ since $\sigma,\tau$ form an iso. In particular, it does directly answer the initial move in $T$. However, by definition
of $\tau^*: T(B) \to T(A)$ it is \emph{strict}, \emph{i.e.} it directly interrogates the left occurrence of $T$, so $\sigma_i$ must necessarily first answer the initial move in $T$ otherwise we would
immediately get a contradiction.

This means that each $\sigma_i$ must first choose a component $j\in J$ (thus inducing a function $f: I \to J$), then play as $\sigma'_i : A_i \tto B_j$. The same analysis on $\tau$ provides a 
function $g: J \to I$ and strategies $\tau'_j: B_j \tto A_{g(j)}$, and it is then obvious that since $\sigma, \tau$ form an isomorphism $g$ must be inverse of $f$ and each $\tau'_j$ inverse of $\sigma'_{g(j)}$.
\end{proof}

\paragraph{Syntactic characterization} Let us prove now that the equational theory $\mathcal{E}$ given in Figure \ref{equational_theory} characterizes the types $A$ and $B$ such that
$\intr{A}$ and $\intr{B}$ are isomorphic in $\Fam(\Gam)_T$.

\begin{lemma}[Type normal form]
Any type $A$ has an unique representative (up to $\iso_\mathcal{E}$) generated by the non-terminal $S$ in:
\begin{eqnarray*}
S&::=& \mathtt{bool}^n \times T\\
T&::=& \mathtt{unit} \mid \Pi_{i\in I} U\\
U&::=& T \to S
\end{eqnarray*}
\end{lemma}
\begin{proof}
Straightforward.
\end{proof}

\begin{proposition}
If $\intr{A}$ and $\intr{B}$ are isomorphic in $\Fam(\Gam)_T$, then $A \iso_{\mathcal{E}} B$.
\label{equational_isos}
\end{proposition}
\begin{proof}
By induction on their normal forms. For types generated by $S$ take $A \iso_{\mathcal{E}} \mathtt{bool}^n \times A'$ and $B \iso_{\mathcal{E}} \mathtt{bool}^p \times B'$.
By Lemma \ref{analysis_isos} we have $n=p$ (since $\intr{A'}$ and $\intr{B'}$, generated by $T$, must be singletons) and we still have $\intr{A'} \iso_{\Fam(\Gam)_T} \intr{B'}$.
The case of types generated by $T$ and $U$ is direct.
\end{proof}

\begin{theorem}
For any types $A, B$ of $\Lang_2$ whose interpretation give families $(A_i)_{i\in I}$ and $(B_j)_{j\in J}$ the following propositions are equivalent:
\begin{itemize}
\item[(1)] $A \iso_{\Lang_2} B$
\item[(2)] $(A_i)_{i\in I} \iso_{\Fam(\Gam)_T/\obeq} (B_j)_{j\in J}$
\item[(3)] $(A_i)_{i\in I} \iso_{\Fam(\Gam)_T} (B_j)_{j\in J}$
\item[(4)] $A \iso_{\mathcal{E}} B$
\end{itemize}
\end{theorem}
\begin{proof}
$(1) \tto (2)$ by soundness and by definition of type isomorphisms in $\Lang_2$,
$(2) \tto (3)$ because the arenas are complete, hence every play of the identity is a prefix of a complete play of the identity, so any isomorphism
$\sigma: A\tto A$ such that $\comp(\sigma) = \comp(id_A)$ must satisfy $id_A \subseteq \sigma$. But as isomorphisms both are total strategies,
so $\sigma = id_A$.
$(3) \tto (4)$ by Proposition \ref{equational_isos}. Finally, $(4) \tto (1)$ because equations in $\mathcal{E}$ can be implemented
in the syntax of $\Lang_2$.
\end{proof}

\subsection{Isomorphisms in the presence of $\mathtt{nat}$}

As suggested by the importance of counting in the proof, the presence of $\mathtt{nat}$ makes it possible to build a non-trivial isomorphism by playing Hilbert's hotel. 
Consider the programming language $\Lang$ from \cite{ahm}, obtained from $\Lang_2$ by replacing $\mathtt{bool}$ with $\mathtt{nat}$. This language has a fully abstract interpretation
in $\BFam(\biggam)_T$, where $\biggam$ is the category of not necessarily finitely branching arenas, and single-threaded strategies.

\begin{proposition}
There is an isomorphism in $\BFam(\biggam)$ between $\intr{(\mathtt{nat} \to \mathtt{unit}) \to (\mathtt{nat} \to \mathtt{unit}) \to \mathtt{unit}}$
and $\intr{(\mathtt{nat}\to \mathtt{unit}) \to (\mathtt{unit} \to \mathtt{unit}) \to \mathtt{unit}}$.
\label{prop_iso}
\end{proposition}
\begin{proof}
By definition of the interpretation of types, this boils down to an isomorphism in $\biggam$ between the two arenas represented in Figure \ref{isoar}. Informally, the isomorphism
from left to right can be described as follows. As long as $q'$ has not been played, it behaves as the identity. Whenever Opponent plays $q'$, it copies it to the other side. Then
if $q'$ has only appeared once, there are two available copies of $\mathtt{nat}\to \mathtt{unit}$ on the left side, one $\mathtt{nat}\to \mathtt{unit}$ and one $\mathtt{unit}\to \mathtt{unit}$ on
the right side, so Player picks a bijection between $\mathbb{N} + \mathbb{N}$ and $\mathbb{N} + 1$ and plays accordingly. More generally, if $q'$ has appeared $n$ times, there are exactly $n+1$ available
copies of $\mathtt{nat}\to \mathtt{unit}$ on the left hand side, one copy of $\mathtt{nat}\to \mathtt{unit}$ and $n$ copies of $\mathtt{unit}\to \mathtt{unit}$ on the right hand side, so Player has to follow
a bijection between $(n+1)\mathbb{N}$ and $\mathbb{N} + n$. So any choice of a bijection between $(n+1)\mathbb{N}$ and $\mathbb{N} + n$ (for all $n\in \mathbb{N}$) will provide an isomorphism.
\end{proof}

\begin{figure}
\[
\xymatrix@C=5pt@R=10pt{
&&q\\
q_1\ar@{-}@/^/[urr]	&q_2\ar@{-}@/^/[ur]\ar@{.}[rr]	&			&a\ar@{-}@/_/[ul]\\
a\ar@{-}[u]		&a\ar@{-}[u]			&			&				&q'\ar@{-}@/_/[ul]\\
			&				&q_1\ar@{-}@/^/[urr]	&q_2\ar@{-}@/^/[ur]\ar@{.}[rr]	&			&a\ar@{-}@/_/[ul]\\
			&				&a\ar@{-}[u]		&a\ar@{-}[u]
}
~~~~
\xymatrix@C=5pt@R=10pt{
&&q\\
q_1\ar@{-}@/^/[urr]     &q_2\ar@{-}@/^/[ur]\ar@{.}[rr]  &                       &a\ar@{-}@/_/[ul]\\
a\ar@{-}[u]             &a\ar@{-}[u]                    &                       &q'\ar@{-}[u]\\
			&				&q\ar@{-}@/^/[ur]	&a\ar@{-}[u]\\
			&				&a\ar@{-}[u]
}
\]
\caption{Non-trivial isomorphic arenas in $\biggam$}
\label{isoar}
\end{figure}

These strategies are not compact so the definability theorem does not apply, however we can nonetheless manually extract corresponding programs from them. We display
them in Figure \ref{coolisos},
where $*$ denotes the product operation on natural numbers, and $\mathtt{div}~M~N$ outputs the result of the division algorithm on $M:\mathtt{nat}$ and $N:\mathtt{nat}$.
Unfortunately, these terms are too complex to hope for a reasonably-sized direct proof that their interpretations give the strategies described above or even that they 
form an isomorphism. This kind of difficulty emphasizes the need for new algebraic methods to manipulate and prove properties of imperative higher-order programs.

\begin{figure*}
{\footnotesize
\begin{minipage}{0.5\linewidth}
\[
\begin{array}{l}
f: (\mathtt{nat} \to \mathtt{unit}) \to (\mathtt{nat} \to \mathtt{unit}) \to \mathtt{unit}\vdash\\
~\mathtt{new}~\mathtt{count}:=0,~\mathtt{func}:= \bot~\mathtt{in}\\
~~\lambda g:\mathtt{nat}\to \mathtt{unit}.\\
~~~\mathtt{let}~x = f~(\lambda n.~g(n*(!\mathtt{count} + 1)))~\mathtt{in}\\
~~~~\lambda h:\mathtt{unit}\to \mathtt{unit}.\\
~~~~~\mathtt{count}~:=~!\mathtt{count} + 1;\\
~~~~~\mathtt{let}~c~=~!\mathtt{count}~\mathtt{in}\\
~~~~~~\mathtt{func}~:=~\mathtt{let}~h=~!\mathtt{func}~\mathtt{in}~\lambda n.~\mathtt{if}~n=~!\mathtt{count}~\mathtt{then}~h~\mathtt{else}~g~n;\\
~~~~~~x~(\lambda n.~\mathtt{if}~n=0~\mathtt{then}~!\mathtt{func}~c~()\\
~~~~~~~~~~~~~~~~~~~~~~~~~~~\mathtt{else}~g((n-1)*(!\mathtt{count} + 1) + c))
\end{array}
\]
\end{minipage}
\begin{minipage}{0.5\linewidth}
\[
\begin{array}{l}
f: (\mathtt{nat} \to \mathtt{unit}) \to (\mathtt{unit} \to \mathtt{unit}) \to \mathtt{unit}\vdash\\
~\mathtt{new}~\mathtt{count}:=0,~\mathtt{func} := \bot~\mathtt{in}\\
~~\lambda g:\mathtt{nat}\to \mathtt{unit}.\\
~~~\mathtt{let}~x=f~(\lambda n.\\
~~~~\mathtt{let}~(q,r)~=~\mathtt{div}~n~(!\mathtt{count}+1)~\mathtt{in}\\
~~~~~\mathtt{if}~r=0~\mathtt{then}~g~q~\mathtt{else}~!\mathtt{func}~r~(q+1))\\
~~~\mathtt{in}\\
~~~~\lambda h:\mathtt{nat}\to \mathtt{unit}.\\
~~~~~\mathtt{count} := !\mathtt{count}+1;\\
~~~~~\mathtt{func} := \mathtt{let}~h=~!\mathtt{func}~\mathtt{in}~\lambda n.~\mathtt{if}~n=~!\mathtt{count}~\mathtt{then}~h~\mathtt{else}~h~n;\\
~~~~~\mathtt{let}~c=!\mathtt{count}~\mathtt{in}~x~(\lambda\_.~!\mathtt{func}~c~0)
\end{array}
\]
\end{minipage}
}
\caption{Type isomorphism in $\Lang$ between $(\mathtt{nat} \to \mathtt{unit}) \to (\mathtt{nat} \to \mathtt{unit}) \to \mathtt{unit}$
and $(\mathtt{nat} \to \mathtt{unit}) \to (\mathtt{unit} \to \mathtt{unit}) \to \mathtt{unit}$.}
\label{coolisos}
\end{figure*}

\section{Conclusion}

We solved Laurent's conjecture and characterized the isomorphisms of types in $\Lang_2$. Surprisingly, we realized that the combination of higher-order
references, natural numbers and call-by-value allowed to define new non-trivial type isomorphisms. Note however that if well-bracketing is satisfied, the proof
of our core game-theoretic theorem adapts directly to arenas where all moves only enable a finite number of questions, but an
arbitrary numbers of answers. As a consequence, there are no non-trivial isomorphisms (\emph{i.e.} not already present in the $\lambda$-calculus) in the call-by-name variant of $\Lang$, although
we can define one using \texttt{call}/\texttt{cc}.

Note that despite the seemingly restricted power of $\Lang_2$, our theorem does apply to all real-life programming languages that have a bounded type of integer, 
such as $\mathtt{bool}^{32}$ or $\mathtt{bool}^{64}$: in this setting, no non-trivial isomorphism can exist. However unbounded natural numbers can be defined using recursive types, so
the isomorphism above can be implemented in a call-by-value programming language with recursive types and general references, such as \textsc{Ocaml}.

This work can be extended in several different ways. An obvious possibility is to study isomorphisms in the presence of sum types, since the model is already equipped
to handle them. We could also try to eliminate bad variables. Murawski and Tzevelekos' games model \cite{DBLP:conf/fossacs/MurawskiT09} of Reduced ML may be a good setting to try that, however
it is not clear whether our core results can be reproved in their nominal setting. 

\textit{Acknowledgments.} We would like to thank Guy McCusker for interesting discussions on the games models for state and for his help to
get a term from the strategy described in Proposition \ref{prop_iso}, and the anonymous referees for their useful comments and suggestions. We 
also would like to acknowledge the support of (UK) EPSRC grant RC-CM1025.


\end{document}